\documentclass[conference]{IEEEtran}

\IEEEoverridecommandlockouts

\usepackage{cite}
\usepackage{amsmath,amssymb,amsfonts}
\usepackage{tikz}
\usetikzlibrary{patterns}
\usetikzlibrary{decorations.pathreplacing}
\usetikzlibrary{shapes.geometric}
\usepackage{pgfplots}
\usetikzlibrary{calc}
\usetikzlibrary{arrows.meta}
\usepackage{amsmath}
\usepackage{amsthm}  
\usepackage{amsfonts}
\usepackage{amssymb}
\usepackage{mathtools}
\usepackage{algorithmic}
\usepackage{graphicx}
\usepackage{textcomp}
\usepackage{blkarray}
\usepackage{array}
\usepackage{multirow}
\usepackage{arydshln}
\usepackage{placeins}
\usepackage{float}
\usepackage{multicol,lipsum,xparse}
\usetikzlibrary{patterns.meta}

\newtheoremstyle{nonitalic}
{3pt}
{3pt}
{\normalfont}
{}
{\bfseries}
{.}
{.5em}
{}

\theoremstyle{nonitalic}

\newtheorem{theorem}{Theorem}

\newtheorem{remark}{Remark}

\usepackage{cite}
\usepackage{amsmath,amssymb,amsfonts}
\usepackage{tikz}
\usetikzlibrary{patterns}
\usetikzlibrary{decorations.pathreplacing}
\usetikzlibrary{matrix,decorations.pathreplacing, calc, positioning}
\usetikzlibrary{shapes.geometric}
\usepackage{pgfplots}
\usetikzlibrary{calc}
\usetikzlibrary{arrows.meta}
\usepackage{amsmath}
\usepackage{amsfonts}
\usepackage{arydshln}
\usepackage{amssymb}
\usepackage{mathtools}
\usepackage{algorithmic}
\usepackage{graphicx}
\usepackage{textcomp}
\usepackage{blkarray}
\usepackage{array}
\usepackage{multirow}
\usetikzlibrary{patterns}

\newcommand\numberthis{\addtocounter{equation}{1}\tag{\theequation}}
\usepackage{xcolor}
\def\BibTeX{{\rm B\kern-.05em{\sc i\kern-.025em b}\kern-.08em
    T\kern-.1667em\lower.7ex\hbox{E}\kern-.125emX}}
    
\makeatletter
\newcommand*{\rom}[1]{\expandafter\@slowromancap\romannumeral #1@}
\makeatother

\begin{document}

\title{\huge Age of Actuated Information and Age of Actuation in a Data-Caching Energy Harvesting Actuator}

\author{
	\IEEEauthorblockN{Ali Nikkhah\IEEEauthorrefmark{1}, 
	Anthony Ephremides\IEEEauthorrefmark{2}, and Nikolaos Pappas\IEEEauthorrefmark{1}\\}
    \IEEEauthorblockA{\IEEEauthorrefmark{1}Department of Computer and Information Science, Link\"{o}ping University, Link\"{o}ping, Sweden}
    \IEEEauthorblockA{\IEEEauthorrefmark{2}Electrical and Computer Engineering, University of Maryland, College Park, MD, USA}
   \{ali.nikkhah, nikolaos.pappas\}@liu.se, etony@umd.edu
   \thanks{This work has been supported in part by the Swedish Research Council (VR), ELLIIT, Zenith, and the European Union (ETHER, 101096526).}
}

\maketitle
\begin{abstract} 
In this paper, we introduce two metrics, namely, \textit{age of actuation (AoA)} and \textit{age of actuated information (AoAI)}, within a discrete-time system model that integrates data caching and energy harvesting (EH). AoA evaluates the timeliness of actions irrespective of the age of the information, while AoAI considers the freshness of the utilized data packet. We use Markov Chain analysis to model the system's evolution. Furthermore, we employ three-dimensional Markov Chain analysis to characterize the stationary distributions for AoA and AoAI and calculate their average values. Our findings from the analysis, validated by simulations, show that while AoAI consistently decreases with increased data and energy packet arrival rates, AoA presents a more complex behavior, with potential increases under conditions of limited data or energy resources. These metrics go towards the semantics of information and goal-oriented communications since they consider the timeliness of utilizing the information to perform an action.
\end{abstract}

\section{Introduction}
Recent advancements in communication semantics \cite{gunduz2022beyond,kountouris2021semantics,lu2023semantics} have marked a pivotal shift from traditional communication paradigms, which primarily focused on the transmission process between a sender and a receiver without considering the contextual significance of the transmitted information. In conventional views, performance metrics remained indifferent to the underlying purpose of the transmitted bits. 

While not inherently contextual, age of information (AoI) \cite{pappas2023age} has been a proxy metric towards semantics of information. AoI quantifies the timeliness of data from a source at its destination, offering a foundation upon which semantic aspects of the source and the destination can be integrated by extending this metric. By semantics, we consider the relevance and importance of information in achieving a specific goal. 

With the proliferation of smaller and autonomous devices, energy harvesting (EH) \cite{sudevalayam2011energy} is considered vital for sustaining the operation of these devices. In this regard, in AoI-related EH works, there are generally two categories of work. The first involves controlled energy sources, such as wireless power transfer systems \cite{nikkhah2023age,nikkhah2023ageo,krikidis2019average,ibrahim2016stability,leng2019ageof}, where the energy source is actively managed within the system design. The second stream focuses on opportunistically harvesting ambient energy, a category our current work belongs to. There are two types of studies within the latter category: continuous-time EH is often represented using a Poisson process \cite{feng2018minimizing, wu2017optimal, YatesISIT2015, zheng2019closed, gindullina2021age, bacingolu2019optimal, elmagid2022age, arafa2019age}, whereas discrete-time EH models typically utilize a Bernoulli process \cite{jia2021age,baknina2018sending, pappas2020average, hatami2021aoi, chen2021optimization, CeranAoI2019,xu2023optimal}, which is the approach of our study.

This paper integrates the semantics of the source and destination into communication system designs by focusing on the timeliness of actions driven by data and energy packets. We want to emphasize that generating more frequent updates will not always lead to better system performance,  as we also demonstrate when considering the entire communication system from information generation to reception and actuation.

In our previous works \cite{nikkhah2023age, nikkhah2023ageo}, we analyzed the timeliness of actions in a system model where actions occur via instantaneous data packets. In this paper, we address scenarios where, upon the arrival of a data packet, when insufficient energy prevents immediate action, the actuator caches the data packet. Consequently, when an action is eventually performed, the triggering data may vary in age, concluding that different actuations will not be identical regarding timing.

This paper extends our previous definition of the \textit{age of actuation (AoA)} to analyze the timeliness of actions, independent of the age of the data packet triggering the action. Furthermore, we propose a new metric, the \textit{Age of Actuated Information (AoAI)}, to evaluate the timeliness of actions concerning the age of the triggering data packets. These two metrics differ in scenarios where caching is possible. In contrast, they would be regarded identical in an instantaneous system model as described in \cite{nikkhah2023age,nikkhah2023ageo}. In this paper, we quantify both metrics for the cacheable non-instantaneous system model and give insights into their similarities and differences.

\section{System Model} \label{System Model}
We consider the system model depicted in Fig. \ref{SystemModelFig}. Time is slotted into equal-length slots. The system includes an actuator at the receiver, which operates using the information from a data packet and the power from an energy packet. A cache is available to store data packets for a later actuation. Additionally, there is a battery capable of storing a single energy packet. A data packet is generated and successfully transmitted at each time slot with a probability of $\lambda_1$. In case the system fails to successfully transmit the generated data, the packet is dropped. As a result, a successfully received data packet is always as fresh as possible. Similarly, with a probability of $\lambda_2$ in each time slot, ambient energy (such as solar, wind, vibration, radio frequency, to name a few) can be harvested to charge the battery with one energy packet. We define the events $\Lambda_1(t)=\{0,1\}$ and $\Lambda_2(t)=\{0,1\}$ as the realizations of $\lambda_1$ and $\lambda_2$, respectively, at time slot $t$. Then, $\Lambda_1(t)=0$ denotes unsuccessful, and $\Lambda_1(t)=1$ denotes a successful data packet reception. Also, $\Lambda_2(t)=0$ denotes unavailability, and $\Lambda_2(t)=1$ denotes availability for harvesting an energy packet.

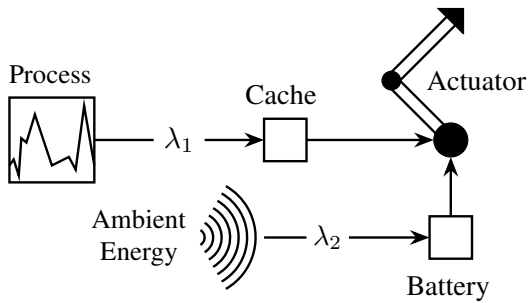
\begin{figure}[t] 
\centering
\resizebox{0.4\textwidth}{!}{%
\begin{tikzpicture}

\draw [thick] (4.5,0) rectangle (5.5,1) node[scale=0.9] [xshift=-0.58cm,yshift=0.3cm] {Process}; 
\draw [thick] (4.5,0.2)-- (4.55,0.27) -- (4.61,0.08) --(4.64,0.52) -- (4.7,0.47) -- (4.8,0.8) -- (5,0.2) -- (5.2,0.3) -- (5.28,0.15) -- (5.38,0.9) -- (5.5,0.2);

\draw [ thick] (5.5,0.5) -- (6.2,0.5) node[xshift=0.3cm, yshift=0cm] {$\lambda_1$};

 \foreach \r in {0.1,0.2,...,0.7} 
    \draw[thick,shift={(6.7cm,-0.65cm)}] (0,0) ++(-60:\r cm) arc (-60:60:\r cm) ;
    \draw[thick] (7.5,-0.65) -- (8,-0.65) node[scale=0.9,align=center, xshift=-2.2cm,yshift=0cm]  {Ambient \\ Energy} ;
    \draw[-Stealth, thick] (8.5,-0.65) -- (9.7-0.25,-0.65) node [ xshift=-1.2cm, yshift=0cm,scale=0.95] {$\lambda_2$};


\draw[-Stealth, thick]  (6.8,0.5) -- (7.5,0.5);


\draw [thick] (7.5,0.75) rectangle (8,0.25) node [ xshift=-0.3cm, yshift=0.78cm,scale=0.95] {Cache};

\draw[-Stealth, thick] (9.5-1.5,1-0.5) -- (11-1.5,1-0.5);
\draw[fill,black] (11.2-1.5,1-0.5) circle (0.2cm);
\draw[thick] (11.25-1.5,1.05-0.5) -- (11.25-0.7-1.5,1.05+0.7-0.5);
\draw[thick] (11.15-1.5,0.95-0.5) -- (11.15-0.7-1.5,0.95+0.7-0.5);
\draw[fill,black] (9,1.70-0.5) circle (0.11cm);
\node [ xshift=10cm, yshift=1.20cm,scale=0.95] {Actuator};
\draw[thick] (10.55-1.5,1.65-0.5) -- (10.55+0.7-1.5,1.65+0.7-0.5);
\draw[thick] (10.45-1.5,1.75-0.5) -- (10.45+0.7-1.5,1.75+0.7-0.5);

\draw[fill=black] (10.55+0.8-1.5,1.65+0.6-0.5) -- (10.45+0.6-1.5,1.75+0.8-0.5) -- (10.55+0.8-1.5,1.75+0.8-0.5)-- (10.55+0.8-1.5,1.65+0.6-0.5);

\draw[thick] (9.7-0.25,-0.4) rectangle (9.7+0.25,-0.9) node [ xshift=-0.28cm, yshift=-0.35cm,scale=0.95] {Battery};

\draw[-Stealth, thick] (9.7,-0.4) -- (9.7,0.3);

\end{tikzpicture}
}
\caption{The considered system model.} 
\label{SystemModelFig}
\end{figure}

The cache suffices to store only one data packet at each time slot. This is because the system is required to actuate upon the freshest data available at each time slot. The actuation is considered to be instantaneous. Consequently, older data packets can be discarded once the most recent data packet is cached. A data packet remains in the cache if it is the freshest available and has not been utilized for actuation. Therefore, upon the reception of a new data packet, the existing one in the cache is replaced. Also, once a data packet is actuated, it is removed from the cache.

Regarding operational sequence within a time slot, the system first checks for the availability of a data packet, whether cached or freshly received in the time slot, then immediately checks the availability of energy, whether in the battery or harvested, and actuates instantaneously, at the beginning of the time slot, if both data and energy are available. The system uses the energy stored in the battery, and if available, it can harvest and store energy for future use at the same time slot. Also, harvesting and utilizing energy for immediate actuation within the same time slot is still possible in scenarios where the battery is depleted

We represent the system state using a two-dimensional Markov process, denoted as $(C(t), B(t))$ to capture the dynamics of the system. In this work, we consider that the cache can store up to one data packet, and the battery can store up to one energy chunk. Thus, the state of the cache is represented by $C(t) \in \{0,1\}$, where $0$ signifies an empty cache and $1$ a full cache at the end of time slot $t$. Similarly, $B(t) \in \{0,1\}$ indicates the battery's state, with $0$ representing an empty battery and $1$ a full battery at the end of time slot $t$. The system can exist in three states: $(0,0)$, $(0,1)$, and $(1,0)$. The state $(1,1)$ is not feasible as if the cache contains an unused data packet and sufficient energy is available, both would be concurrently utilized for an actuation. Consequently, the transition probability matrix for this process is defined as
\setcounter{MaxMatrixCols}{30}
\begin{equation*} \label{TPM_SystemModel}
\mathbf{P}_s=
\begin{bmatrix}
\lambda_1 \lambda_2 + \bar{\lambda}_1 \bar{\lambda}_2 &  \bar{\lambda}_1 \lambda_2 & \lambda_1 \bar{\lambda}_2\\
\lambda_1 \bar{\lambda}_2 & \lambda_2 +\bar{\lambda}_1 \bar{\lambda}_2 &0 \\
 \lambda_2  & 0 & \bar{\lambda}_2  \\
\end{bmatrix} ,
\end{equation*}

where the rows and columns represent the joint states $(C(t),B(t))$ for states' realizations $(0,0)$, $(0,1)$, and $(1,0)$, respectively.

\section{Analytical Results}
In this section, we present the analysis for the AoA and the AoAI for the considered system model.

\subsection{Age of Information (AoI)}

AoI represents the time elapsed since the generation of the freshest data packet available. It is defined as $I(t)=t-u(t)$, where $u(t)$ is the timestamp of the last successfully received data packet. Accordingly, AoI is reset to $1$ with probability $\lambda_1$ and increments by $1$ with probability $\bar{\lambda}_1$. 
Then, the average AoI would be $\bar{I}=\frac{1}{\lambda_1}$.

\begin{figure}[!h]
\centering 
\scalebox{0.9}{ \boldmath{
\begin{tikzpicture}

{
    \pgfdeclarepatternformonly{Sparse Vertical Lines}
    {%
        \pgfqpoint{-1pt}{-1pt}%
    }
    {%
        \pgfqpoint{10pt}{10pt}%
    }
    {%
        \pgfqpoint{9pt}{9pt}%
    }
    {
        \pgfsetlinewidth{0.1pt} 
        \pgfpathmoveto{\pgfqpoint{0pt}{0pt}}
        \pgfpathlineto{\pgfqpoint{0pt}{9.1pt}}
        \pgfusepath{stroke}
    }
}

{
    \pgfdeclarepatternformonly{Sparse Horizontal Lines}
    {%
        \pgfqpoint{-1pt}{-1pt}%
    }
    {%
        \pgfqpoint{10pt}{10pt}%
    }
    {%
        \pgfqpoint{9pt}{9pt}%
    }
    {
        \pgfsetlinewidth{0.1pt} 
        \pgfpathmoveto{\pgfqpoint{0pt}{0pt}}
        \pgfpathlineto{\pgfqpoint{9.1pt}{0pt}}
        \pgfusepath{stroke}
    }
    }

{
    \pgfdeclarepatternformonly{Sparse North East Lines}
    {%
        \pgfqpoint{-1pt}{-1pt}%
    }
    {%
        \pgfqpoint{10pt}{10pt}%
    }
    {%
        \pgfqpoint{9pt}{9pt}%
    }
    {
        \pgfsetlinewidth{0.1pt} 
        \pgfpathmoveto{\pgfqpoint{0pt}{0pt}}
        \pgfpathlineto{\pgfqpoint{9.1pt}{9.1pt}}
        \pgfusepath{stroke}
    }
    }

    \pgfdeclarepatternformonly{Vertical Lines}
    {%
        \pgfqpoint{-1pt}{-1pt}%
    }
    {%
        \pgfqpoint{4pt}{4pt}%
    }
    {%
        \pgfqpoint{3pt}{3pt}%
    }
    {
        \pgfsetlinewidth{0.1pt} 
        \pgfpathmoveto{\pgfqpoint{0pt}{0pt}}
        \pgfpathlineto{\pgfqpoint{0pt}{3.1pt}}
        \pgfusepath{stroke}
    }

    \pgfdeclarepatternformonly{Horizontal Lines}
    {%
        \pgfqpoint{-1pt}{-1pt}%
    }
    {%
        \pgfqpoint{4pt}{4pt}%
    }
    {%
        \pgfqpoint{3pt}{3pt}%
    }
    {
        \pgfsetlinewidth{0.1pt} 
        \pgfpathmoveto{\pgfqpoint{0pt}{0pt}}
        \pgfpathlineto{\pgfqpoint{3.1pt}{0pt}}
        \pgfusepath{stroke}
    }

    \pgfdeclarepatternformonly{North East Lines}
    {%
        \pgfqpoint{-1pt}{-1pt}%
    }
    {%
        \pgfqpoint{4pt}{4pt}%
    }
    {%
        \pgfqpoint{3pt}{3pt}%
    }
    {
        \pgfsetlinewidth{0.1pt} 
        \pgfpathmoveto{\pgfqpoint{0pt}{0pt}}
        \pgfpathlineto{\pgfqpoint{3.1pt}{3.1pt}}
        \pgfusepath{stroke}
    }

\fill[pattern=Sparse Vertical Lines] (0,0) rectangle (1,1);
\fill[pattern=Sparse Vertical Lines] (1,0) rectangle (2,2);
\fill[pattern=Sparse Vertical Lines] (2,0) rectangle (3,1);
\fill[pattern=Sparse Vertical Lines] (3,0) rectangle (4,2);
\fill[pattern=Sparse Vertical Lines] (4,0) rectangle (5,3);
\fill[pattern=Sparse Vertical Lines] (5,0) rectangle (6,4);
\fill[pattern=Sparse Vertical Lines] (6,0) rectangle (7,1);
\fill[pattern=Sparse Vertical Lines] (7,0) rectangle (8,2);

\draw[line width=4pt, green]  (0,1) -- (1,1) -- (1,2) -- (2,2) -- (2,1) -- (3,1) -- (3,2) -- (4,2) -- (4,3) -- (5,3) -- (5,4) -- (6,4) -- (6,1) -- (7,1) -- (7,2) -- (8,2);

\fill[pattern=Sparse Horizontal Lines] (0,0) rectangle (1,1);
\fill[pattern=Sparse Horizontal Lines] (1,0) rectangle (2,2);
\fill[pattern=Sparse Horizontal Lines] (2,0) rectangle (3,3);
\fill[pattern=Sparse Horizontal Lines] (3,0) rectangle (4,4);
\fill[pattern=Sparse Horizontal Lines] (4,0) rectangle (5,1);
\fill[pattern=Sparse Horizontal Lines] (5,0) rectangle (6,2);
\fill[pattern=Sparse Horizontal Lines] (6,0) rectangle (7,3);
\fill[pattern=Sparse Horizontal Lines] (7,0) rectangle (8,4);

\draw[line width=2pt, blue]  (0,1) -- (1,1) -- (1,2) -- (2,2) -- (2,3) -- (3,3) -- (3,4) -- (4,4) -- (4,1) -- (5,1) -- (5,2) -- (6,2) -- (6,3) -- (7,3) -- (7,4) -- (8,4);

\fill[pattern=Sparse North East Lines] (0,0) rectangle (1,1);
\fill[pattern=Sparse North East Lines] (1,0) rectangle (2,2);
\fill[pattern=Sparse North East Lines] (2,0) rectangle (3,3);
\fill[pattern=Sparse North East Lines] (3,0) rectangle (4,4);
\fill[pattern=Sparse North East Lines] (4,0) rectangle (5,3);
\fill[pattern=Sparse North East Lines] (5,0) rectangle (6,4);
\fill[pattern=Sparse North East Lines] (6,0) rectangle (7,5);
\fill[pattern=Sparse North East Lines] (7,0) rectangle (8,6);

\draw[line width=0.8pt, red] (0,1) -- (1,1) -- (1,2) -- (2,2) -- (2,3) -- (3,3) -- (3,4) -- (4,4) -- (4,3) -- (5,3) -- (5,4) -- (6,4) -- (6,5) -- (7,5) -- (7,6) -- (8,6);

\draw[-Stealth, very thick ] (0,0) -- (9,0)  node [xshift=-0.2cm, yshift=0.4cm] {$t$} node [xshift=-9cm, yshift=-0.4cm] {$0$}  node [xshift=-8cm, yshift=-0.4cm] {$1$}  node [xshift=-7cm, yshift=-0.4cm] {$2$}  node [xshift=-6cm, yshift=-0.4cm] {$3$}  node [xshift=-5cm, yshift=-0.4cm] {$4$}  node [xshift=-4cm, yshift=-0.4cm] {$5$}  node [xshift=-3cm, yshift=-0.4cm] {$6$}  node [xshift=-2cm, yshift=-0.4cm] {$7$}  node [xshift=-1cm, yshift=-0.4cm] {$8$};

\draw[-Stealth, very thick ] (0,0) -- (0,7) node [xshift=0.7cm, yshift=-0.2cm] {$AI(t)$}  node [xshift=0.7cm, yshift=-0.7cm] {$A(t)$} node [xshift=0.7cm, yshift=-1.2cm] {$I(t)$} node [xshift=-0.3cm, yshift=-6cm] {$1$} node [xshift=-0.3cm, yshift=-7cm] {$0$};

\draw (1,-0.1) -- (1,0.1);
\draw (2,-0.1) -- (2,0.1);
\draw (3,-0.1) -- (3,0.1);
\draw (4,-0.1) -- (4,0.1);
\draw (5,-0.1) -- (5,0.1);
\draw (6,-0.1) -- (6,0.1);
\draw (7,-0.1) -- (7,0.1);
\draw (8,-0.1) -- (8,0.1);

\draw (-0.1,1) -- (0.1,1);
\draw (-0.1,2) -- (0.1,2);
\draw (-0.1,3) -- (0.1,3);
\draw (-0.1,4) -- (0.1,4);
\draw (-0.1,5) -- (0.1,5);
\draw (-0.1,6) -- (0.1,6);

\node[xshift=3cm,yshift=6.5cm] {AoAI};
\node[xshift=3cm,yshift=6cm] {AoA};
\node[xshift=3cm,yshift=5.5cm] {AoI};

\fill[pattern=North East Lines] (3.5,6.3) rectangle (4.5,6.7);
\fill[pattern=Horizontal Lines] (3.5,5.8) rectangle (4.5,6.2);
\fill[pattern=Vertical Lines] (3.5,5.3) rectangle (4.5,5.7);

\draw[line width = 2pt, red] (3.5,6.7) -- (4.5,6.7);
\draw[line width = 2pt, blue]  (3.5,6.2) rectangle (4.5,6.2);
\draw[line width = 2pt, green] (3.5,5.7) rectangle (4.5,5.7);

\draw (2.5,5.2) rectangle (4.7,6.8);

\end{tikzpicture}}} 
\caption{A sample path of the evolution of AoAI, AoA, and AoI.}
\label{AoAI_AoA_AoI_Evo}
\end{figure}
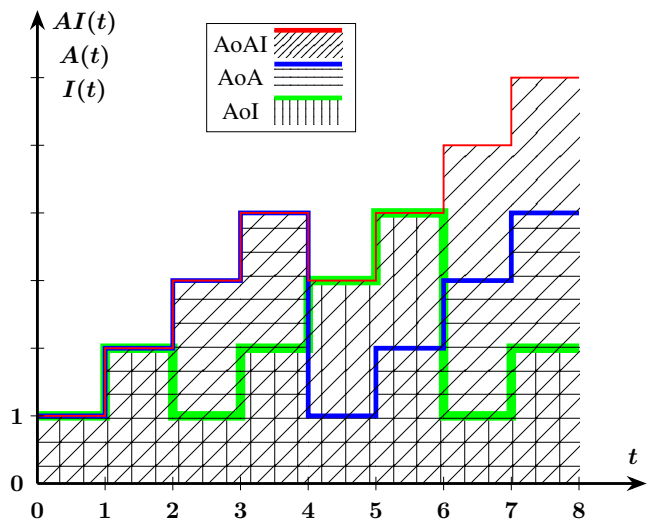

\begin{table}[!h]
\centering
\caption{Events of data reception and energy availability, and states of the cache and the battery.}
\label{Events}
\begin{tabular}{c|c|c|c|c|c|c|c|}
\cline{2-8}
& $t=1$ & $t=2$ & $t=3$ & $t=4$ & $t=5$ & $t=6$ & $t=7$ \\ \hline
\multicolumn{1}{|c|}{$\Lambda_1(t)$}  &0 &1 &0 &0 &0 &1 &0 \\ \hline
\multicolumn{1}{|c|}{$\Lambda_2(t)$} &0 &0 &0 &1 &0 &0 &0 \\ \hline
\multicolumn{1}{|c|}{$C(t)$} &0 &1 &1 &0 &0 &1 &1 \\ \hline
\multicolumn{1}{|c|}{$B(t)$} &0 &0 &0 &0 &0 &0 &0 \\ \hline
\end{tabular}
\end{table}

\subsection{Age of Actuation (AoA)} 
AoA is defined as $A(t)=t-a(t)$, where $a(t)$ is the timestamp of the last performed actuation. At each time slot, AoA is reset to $1$ if there is an actuation and increases by $1$ otherwise, irrespective of the age of the actuated data packet. This metric is relevant for capturing the timeliness of actions without considering the freshness of the actuated information.

\begin{figure*}
\hrule
\vspace{1em}
\begin{equation} \label{Average_AoA}
\bar{A}=
\frac{\lambda_{1}^{4} \left(\lambda_{2} - 1\right)^{3} - 2 \lambda_{1}^{3} \lambda_{2} \left(\lambda_{2} - 1\right)^{2} - \lambda_{1}^{2} \lambda_{2}^{2} \left(2 \lambda_{2}^{2} - 3 \lambda_{2} + 1\right) + \lambda_{1} \lambda_{2}^{3} \left(3 \lambda_{2} - 2\right) - \lambda_{2}^{4}}{\lambda_{1} \lambda_{2} \left(\lambda_{1} \left(\lambda_{2} - 1\right) - \lambda_{2}\right) \left(\lambda_{1}^{2} \left(\lambda_{2} - 1\right)^{2} + \lambda_{1} \left(- 2 \lambda_{2}^{2} + \lambda_{2}\right) + \lambda_{2}^{2}\right)}
\vspace{-5pt}
\end{equation}
\end{figure*}
\begin{figure*}
{\small
\begin{equation} \label{Average_AoAI}
\overline{AI}=\frac{\lambda_1^4 \left(\lambda_2-4\right) \left(\lambda_2-1\right)^3 \lambda_2 - 4 \lambda_1^3 \left( \lambda_2-1\right)^3 \lambda_2^2 + \lambda_1 \left(3 - 4 \lambda_2\right) \lambda_2^4 + \lambda_2^5 + \lambda_1^5 \left(\lambda_2-1\right)^3 \left(2 \lambda_2-1\right) + 2 \lambda_1^2 \lambda_2^3 \left(2 - 5 \lambda_2 + 3 \lambda_2^2\right)}{\lambda_1 \lambda_2 \left(\lambda_1 + \lambda_2 - \lambda_1 \lambda_2\right)^2 \left(\lambda_1^2 \left(\lambda_2-1\right)^2 + \lambda_2^2 + \lambda_1 \left(\lambda_2 - 2 \lambda_2^2\right)\right)}
\end{equation}}
\hrule
\end{figure*}

\begin{theorem}         
The average AoA is given by (\ref{Average_AoA}) at the top of this page.
\end{theorem}

\begin{proof}
Refer to Appendix \ref{Proof of the Theorem 1: The Average AoA}.
\end{proof}

\subsection{Age of Actuated Information (AoAI)}
AoAI is defined as $AI(t)=t-a(t)+I(a(t))=A(t)+I(a(t))$, where $I(a(t))$ represents the AoI at the time of the last actuation. AoAI measures the time elapsed since the generation of the last actuated data packet. AoAI is reset to the current AoI upon an actuation and increases by $1$ otherwise. This metric assesses the timeliness of actions in relation to the freshness of the actuated data packets.
\begin{theorem}
The average AoAI is given by (\ref{Average_AoAI}) at the top of this page.
\end{theorem}

\begin{proof}
See Appendix \ref{Proof of the Theorem 2: The Average AoAI}.
\end{proof}

\begin{figure}
    \centering
    \includegraphics[scale=0.6]{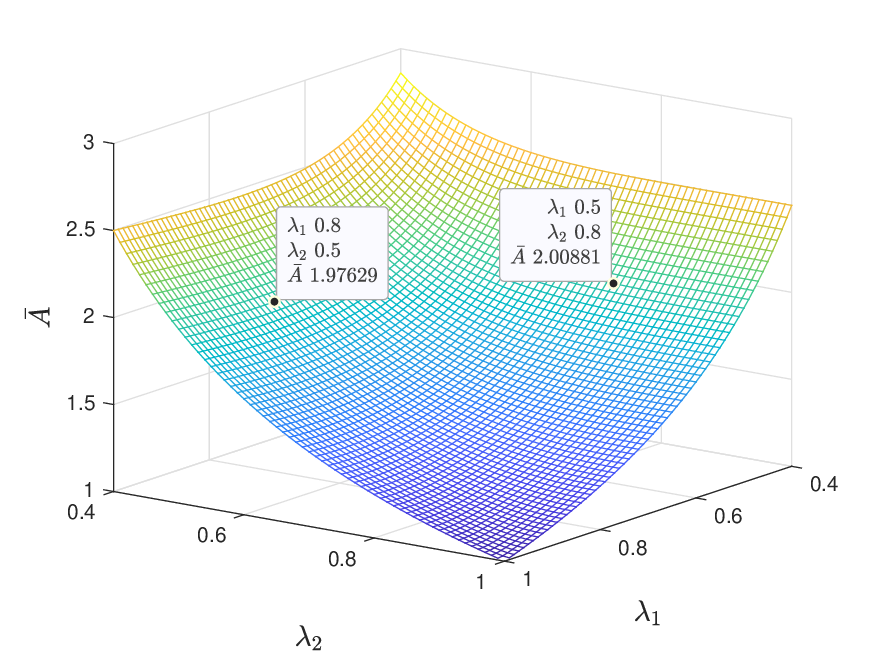}
    \caption{The average AoA versus $\lambda_1$ and $\lambda_2$.}
     \label{agesurf}
\end{figure}

\begin{figure}
    \centering
    \includegraphics[scale=0.6]{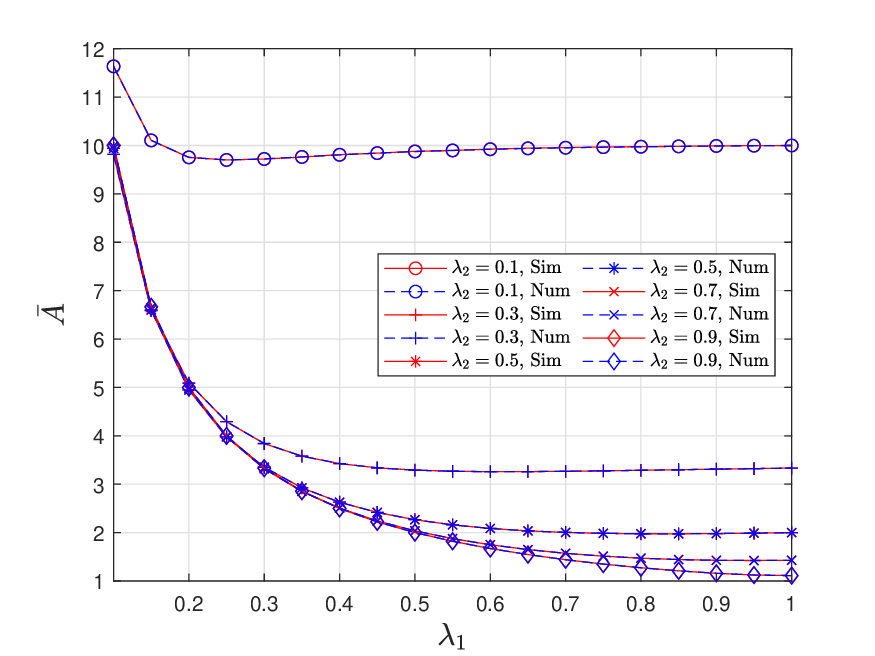}
    \caption{The average AoA for the different values of $\lambda_2$ versus $\lambda_1$.}
    \label{AoAq2_0.13579}
\end{figure}

\begin{figure}
    \centering
    \includegraphics[scale=0.6]{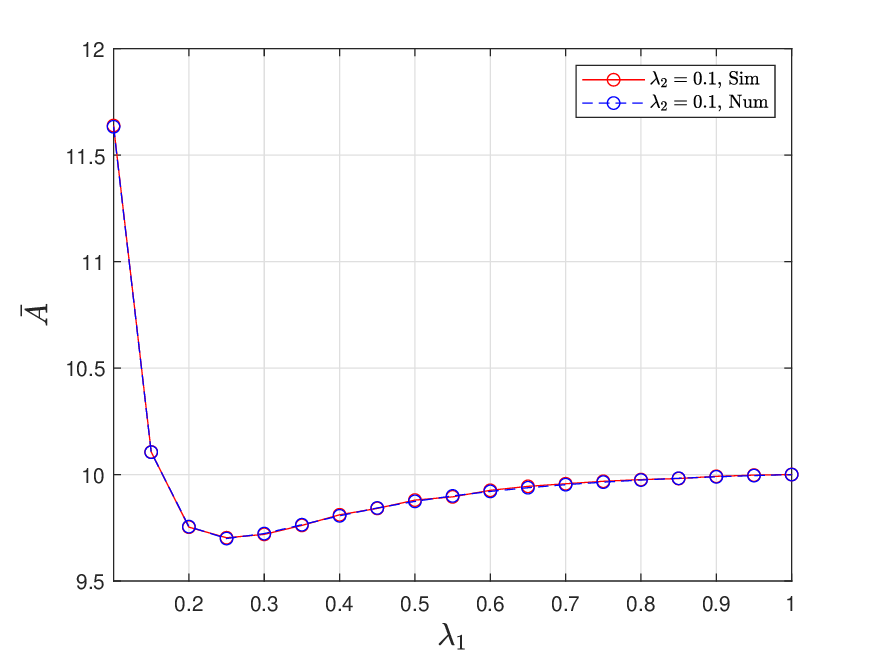}
    \caption{The average AoA for $\lambda_2=0.1$ versus $\lambda_1$.}
     \label{AoAq2_0.1}
\end{figure}

\subsection{The relations between the AoI, AoA, and AoAI}
The two metrics of AoA and AoAI can be considered the same in an instantaneous system model in \cite{nikkhah2023age,nikkhah2023ageo}, since $I(a(t))=0$ for any actuation when the packet is fresh at the beginning of the actuation time slot.
Fig. \ref{AoAI_AoA_AoI_Evo} illustrates a sample path demonstrating the concurrent evolution of these three metrics. The realizations of the random variables $\lambda_1$ and $\lambda_2$, and the processes $C(t)$ and $B(t)$ are also presented for the illustrated time slots in table \ref{Events}.

We have data reception in time slot $t=2$ and $t=6$. Thus, AoI is reset to $1$. In time slot $t=4$, an energy packet is harvested, and the system performs an actuation based on a data packet received and cached at time slot $t=2$. Thus, at time slot $t=4$, the AoI keeps growing, AoA is reset to $1$, and AoAI is reset to the AoI since the AoI is the age of the actuated packet at this time slot.

\vspace{10pt}
\begin{remark}
We always have $AI(t) \geq I(t)$ and $A(t) \geq A(t)$: AoAI is an upper bound for both the AoI and AoA. However, it is evident by the AoI definition that the reception of a data packet is always prior to its utilization for an actuation. For AoA, it is because the age of a packet cannot be lower than $1$ at the end of a time slot. As a result then $\overline{AI} \geq  \bar{I}$ and $\overline{AI} \geq \bar{A}$. Note that the AoA can be less or larger than the AoI, as both cases can be noticed and tracked in Fig. \ref{AoAI_AoA_AoI_Evo}. Nevertheless, ultimately, we always have $\bar{I} \leq \bar{A}$. The intuition is that each successfully received data packet resets the AoI to $1$ but does not necessarily reset the AoA to $1$. In other words, for the AoI to reset to $1$, only a received data packet is needed. However, an energy packet is also needed for the AoA to reset to $1$. Therefore, we have $\bar{I} \leq \bar{A} \leq \overline{AI}$.
Also, in extreme cases, we have:
\begin{itemize}
    \item $\lambda_1 \rightarrow 1$ : \  $\overline{AI} \rightarrow \bar{A} \rightarrow \displaystyle{\frac{1}{\lambda_2}} \geq \bar{I} \rightarrow 1 $,
    \item $\lambda_2 \rightarrow 1$ : \  $\overline{AI} \rightarrow \bar{A} \rightarrow \bar{I} \rightarrow \displaystyle{\frac{1}{\lambda_1}} \geq 1 $,
    \item $\lambda_1 \rightarrow 1, \ \lambda_2 \rightarrow 1$: \ $\overline{AI} \rightarrow \bar{A} \rightarrow \bar{I} \rightarrow  1 $.
\end{itemize}
\end{remark}
\section{Numerical and Simulation Results}

In this section, we present the numerical and simulation results for the average AoA and AoAI. We observe that both the average AoA and average AoAI exhibit approximately symmetric behavior with respect to their variables $\lambda_1$ and $\lambda_2$. It means that if we define $\bar{A}(\lambda_1,\lambda_2)$ and $\overline{AI}(\lambda_1,\lambda_2)$, we find that $\bar{A}(\lambda_1,\lambda_2) \approx \bar{A}(\lambda_2,\lambda_1)$ and $\overline{AI}(\lambda_1,\lambda_2) \approx \overline{AI}(\lambda_2,\lambda_1)$. Given the symmetry, we will only show present $\bar{A}$ and $\overline{AI}$ for various values of $\lambda_2$ against $\lambda_1$, as the reverse will not give more insights. 

\begin{figure}
    \centering
    \includegraphics[scale=0.6]{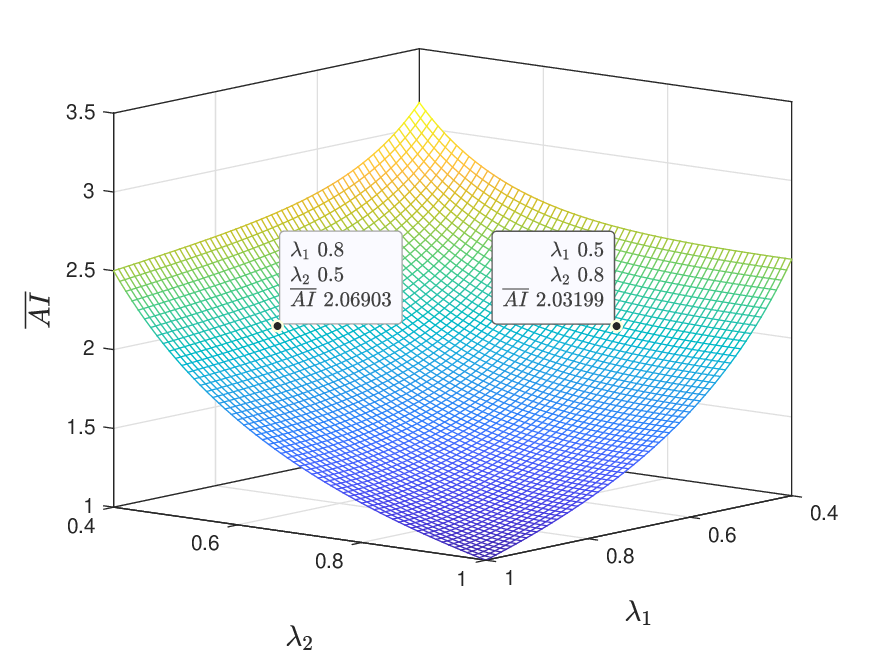}
    \caption{The average AoA versus $\lambda_1$ and $\lambda_2$.}
     \label{ageinfsurf_num}
\end{figure}

\begin{figure} 
    \centering
    \includegraphics[scale=0.6]{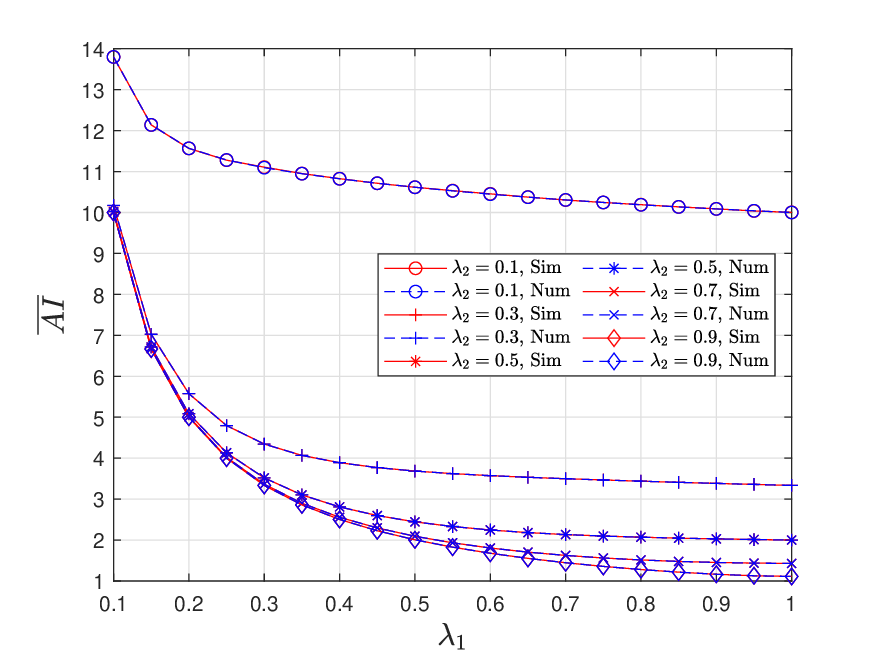}
    \caption{The average AoA for the different values of $\lambda_2$ versus $\lambda_1$.}
    \label{AoAIq2_0.13579}
\end{figure}

Fig. \ref{agesurf} shows the results for the average AoA against $\lambda_1$ and $\lambda_2$. It also can help show the symmetry mentioned above for the average AoA. In Fig. \ref{AoAq2_0.13579}, the curves for average AoA for five different values of $\lambda_2$ against $\lambda_1$ are depicted. Fig. \ref{AoAq2_0.1} displays the curve corresponding to $\lambda_2=0.1$ from Fig. \ref{AoAq2_0.13579}. Notably, for lower values of $\lambda_2$ ($\lambda_1$), increasing $\lambda_1$ ($\lambda_2$) leads to a rise in the average AoA, which is counterintuitive at first. This unexpected phenomenon provides the following insight: updating too frequently can result in higher average AoA values when data availability is very limited. Therefore, less frequent updates might be more effective in reducing the timeliness of actions when the age of the actuated packets is not a critical factor.

Fig. \ref{ageinfsurf_num} shows results for the average AoAI against $\lambda_1$ and $\lambda_2$. It also illustrates the previously mentioned symmetry for the average AoAI. Fig \ref{AoAIq2_0.13579} illustrates the average AoAI for different values of $\lambda_2$ against $\lambda_1$. Contrary to the average AoA, the average AoAI does not exhibit the same unexpected behavior. This means that, even at low values of $\lambda_2$ ($\lambda_1$), increasing $\lambda_1$ ($\lambda_2$) decreases the average AoAI, showcasing a different dynamic than the average AoA.

\section{Conclusion}
In this study, we have introduced two metrics to assess the timeliness of actions in systems where data packet caching is possible. The first metric, Age of Actuation (AoA), captures an action's timeliness irrespective of the packet's age. The second, the Age of Actuated Information (AoAI), incorporates the freshness of the utilized data packet into the evaluation. We formulated these metrics within a three-dimensional Markov framework to determine their steady states and calculated the average AoA and AoAI. Our numerical results, validated by simulations, illustrate the average AoA and AoAI performance under various conditions.

We illustrated that the average AoAI is strictly decreasing against $\lambda_1$ and $\lambda_2$. In contrast, the average AoA exhibits a more complex behavior: when data or energy arrivals are infrequent, increasing the frequency of the other can lead to suboptimal utilization of the scarce resource, thereby increasing the AoA.

\begin{appendices}
\section{Proof of the Theorem 1: The Average AoA} \label{Proof of the Theorem 1: The Average AoA}
To have the AoA states as a Markov process, consider the three-dimensional Markov process $(A,C,B)$, in which $A \in \{1,2,3,\hdots\}$ is the AoA at the end of the time slot. The state of the cache is represented by $C(t) \in {0,1}$, where $0$ signifies an empty cache and $1$ a full cache at the end of time slot $t$. Similarly, $B(t) \in {0,1}$ indicates the battery's state, with $0$ representing an empty battery and $1$ a full battery at the end of time slot $t$.

Thus, the states are $\{1,0,0\}$, $\{1,0,1\}$, $\{2,0,0\}$, $\{2,0,1\}$, $\{2,1,0\}$, $\{3,0,0\}$, $\{3,0,1\}$, $\{3,1,0\}$, and so on. Obviously, we cannot have states such as $\{A,1,1\}$. Because as we mentioned earlier, if there are both energy packet and data packet, they should already have been utilized for an actuation together. Also, we cannot have $\{1,1,0\}$, since $A=1$ means an instantaneous data packet has been utilized, and the cache should be empty.
Therefore, the transition probability matrix for this process can be represented in (\ref{TPM_AoA}).

\begin{figure*}
\hrule
\vspace{2em}
\setcounter{MaxMatrixCols}{30}
\begin{equation} \label{TPM_AoA}
\mathbf{P}^{A}=
\begin{bmatrix}
\lambda_1 \lambda_2 & 0  & \bar{\lambda}_1 \bar{\lambda}_2 & \bar{\lambda}_1 \lambda_2 & \lambda_1 \bar{\lambda}_2&0& 0& 0&0& 0&0& 0&\hdots\\
\lambda_1 \bar{\lambda}_2 & \lambda_1 \lambda_2 &0 &\bar{\lambda}_1 &  0 & 0&0&0&0& 0& 0& 0& \hdots \\
\lambda_1 \lambda_2 & 0 &0 & 0 &0 & \bar{\lambda}_1  \bar{\lambda}_2 & \bar{\lambda}_1 \lambda_2 & \lambda_1  \bar{\lambda}_2 &0& 0& 0& 0& \hdots \\
\lambda_1 \bar{\lambda}_2 & \lambda_1 \lambda_2 & 0&0& 0& 0 & \bar{\lambda}_1  &0& 0&0& 0& 0  & \hdots \\
\lambda_2 & 0&0&0&0&0&0& \bar{\lambda}_2& 0&0& 0& 0  & \hdots\\
\lambda_1 \lambda_2 &0&0&0&0&0&0&0&\bar{\lambda}_1  \bar{\lambda}_2 & \bar{\lambda}_1 \lambda_2 & \lambda_1  \bar{\lambda}_2& 0 & \hdots\\
\lambda_1 \bar{\lambda}_2 & \lambda_1 \lambda_2 & 0& 0& 0&0&0& 0& 0 & \bar{\lambda}_1  &0& 0  & \hdots\\
\lambda_2 & 0&0&0&0&0&0&0&0&0& \bar{\lambda}_2& 0 & \hdots\\
\vdots & \vdots & \vdots&\vdots&\vdots&\vdots&\vdots&\vdots&\vdots&\vdots&\vdots&\vdots & \ddots\\
\end{bmatrix} .
\end{equation}
\end{figure*}

\begin{figure*}
\setcounter{MaxMatrixCols}{30}
\begin{equation} \label{TPM_AoAI}
\mathbf{P}^{AI}=
\begin{bmatrix}
\lambda_1 \lambda_2 & 0  & \lambda_1 \bar{\lambda}_2 & \bar{\lambda}_1 \bar{\lambda}_2 & \bar{\lambda}_1 \lambda_2&0& 0& 0&0& 0&0& 0&0&0&0&\hdots\\

\lambda_1 \bar{\lambda}_2 & \lambda_1 \lambda_2 &0&0 &\bar{\lambda}_1 &  0 & 0&0&0&0&0&0&0& 0&  0& \hdots \\

\lambda_1 \lambda_2 &0 &0& \bar{\lambda}_1 \lambda_2 &0 &\lambda_1  \bar{\lambda}_2& \bar{\lambda}_1  \bar{\lambda}_2   &0& 0& 0& 0& 0& 0& 0&  0& \hdots \\

\lambda_1 \lambda_2 & 0 &0&0 &0&\lambda_1  \bar{\lambda}_2& 0 & \bar{\lambda}_1  \bar{\lambda}_2& \bar{\lambda}_1 \lambda_2    &0& 0&0& 0& 0& 0& \hdots \\

\lambda_1 \bar{\lambda}_2 & \lambda_1 \lambda_2&	0&	0&	0&	0&	0&	0&	\bar{\lambda}_1&	0&	0&	0 &0&  0& 0& \hdots\\

\lambda_1 \lambda_2 &0 &0& \bar{\lambda}_1 \lambda_2 &0&0& 0& 0& 0 &\lambda_1  \bar{\lambda}_2& \bar{\lambda}_1  \bar{\lambda}_2   &0&0&  0& 0& \hdots \\

\lambda_1 \lambda_2 & 0 &0& 0 &0&0 &0&\bar{\lambda}_1  \lambda_2& 0 & \lambda_1  \bar{\lambda}_2&0& \bar{\lambda}_1  \bar{\lambda}_2    &0&  0& 0& \hdots \\

\lambda_1 \lambda_2 & 0 &0& 0 &0&0 &0&0&0&\lambda_1  \bar{\lambda}_2& 0&0 & \bar{\lambda}_1  \bar{\lambda}_2& \bar{\lambda}_1 \lambda_2    & 0& \hdots \\

\lambda_1 \bar{\lambda}_2 & \lambda_1 \lambda_2&	0&	0&	0&	0&	0& 0&	0&	0&	0&	0&	0&	\bar{\lambda}_1&	0&	\hdots \\

\vdots&\vdots&\vdots&\vdots&\vdots&\vdots&\vdots&\vdots&\vdots&\vdots&\vdots&\vdots&\vdots&\vdots&\vdots & \ddots\\ 
\end{bmatrix}
\end{equation}
\hrule
\end{figure*}

The probability of occurrence of the state $(A,C,B)$ is denoted by $V_{A,C,B}$. The steady states vector is $\mathbf{V} = [V_{1,0,0}, V_{1,0,1}, V_{2,0,0}, V_{2,0,1}, V_{2,1,0}, \hdots]$. Hence, by the equations $\mathbf{V}\mathbf{P}=\mathbf{V}$ and $\mathbf{V} \mathbf{1}=1$, we can obtain the steady states. We define $\Delta_{j,k}=\sum_{i=1}^{\infty}V_{i,j,k}$.
The first two equations of $\mathbf{V}\mathbf{P}=\mathbf{V}$ yield
\begin{align}
V_{1,0,0}&=\lambda_1 \lambda_2 \Delta_{0,0} + \lambda_1 \bar{\lambda}_2 \Delta_{0,1} + \lambda_2 \Delta_{1,0},\\
V_{1,0,1}&=\lambda_1 \lambda_2 \Delta_{0,1}.  \label{V_101_Delta_01}
\end{align}
Writing the other equations of $\mathbf{V}\mathbf{P}=\mathbf{V}$ can provide us with three recursive structural equations
\begin{align}
V_{i+1,0,0}&=\bar{\lambda}_1 \bar{\lambda}_2  V_{i,0,0}, \  \forall i \in \{1,2,\hdots\},\\
V_{i+1,0,1}&=\bar{\lambda}_1 \lambda_2  V_{i,0,0} + \bar{\lambda}_1 V_{i,0,1}, \  \forall i \in \{1,2,\hdots\},\\
V_{i+1,1,0}&=\lambda_1 \bar{\lambda}_2   V_{i,0,0} + \bar{\lambda}_2 V_{i,1,0}, \  \forall i \in \{1,2,\hdots\}.
\end{align}
Taking $w=\lambda_1 \lambda_2$, $x=\lambda_1 \bar{\lambda}_2$, $y=\bar{\lambda}_1 \lambda_2$, and $z=\bar{\lambda}_1 \bar{\lambda}_2$, by the recursive equations, we can write all of the state probabilities based on only $V_{1,0,0}$ and $V_{1,0,1}$:
\begin{align*}
V_{1,0,0}&=V_{1,0,0}\\
V_{1,0,1}&=V_{1,0,1}\\
V_{2,0,0}&=(z) V_{1,0,0}\\
V_{2,0,1}&=(y) V_{1,0,0} + (\bar{\lambda}_1) V_{1,0,1}\\
V_{2,1,0}&=(x) V_{1,0,0}\\
V_{3,0,0}&=(z^2) V_{1,0,0}\\
V_{3,0,1}&=(zy+\bar{\lambda}_1 y) V_{1,0,0} + (\bar{\lambda}_1^2) V_{1,0,1}\\
V_{3,1,0}&=(zx+\bar{\lambda}_2 x) V_{1,0,0}\\
V_{4,0,0}&=(z^3) V_{1,0,0}\\
V_{4,0,1}&=(yz^2+\bar{\lambda}_1  yz + \bar{\lambda}_1^2 y) V_{1,0,0} + (\bar{\lambda}_1^3) V_{1,0,1}\\
V_{4,1,0}&=(xz^2+\bar{\lambda}_2 xz+\bar{\lambda}_2 x) V_{1,0,0}\\
V_{5,0,0}&=(z^4) V_{1,0,0}\\
V_{5,0,1}&=(yz^3+\bar{\lambda}_1 y z^2 + \bar{\lambda}_1^2 yz+\bar{\lambda}_1^3 y) V_{1,0,0} + (\bar{\lambda}_1^4) V_{1,0,1}\\
V_{5,1,0}&=(xz^3+\bar{\lambda}_2 xz^2+\bar{\lambda}_2^2 xz+ \bar{\lambda}_2^3 x) V_{1,0,0} \numberthis \label{StateEquations_AoA}
\end{align*}
If we sum the corresponding terms, we obtain:
\begin{equation}
\Delta_{0,0}=\frac{1}{1-z} V_{1,0,0},
\end{equation}
\begin{equation} \label{Delta_01}
\Delta_{0,1}=\frac{y}{\lambda_1} \frac{1}{1-z} V_{1,0,0} + \frac{1}{\lambda_1} V_{1,0,1},
\end{equation}
\begin{equation}
\Delta_{1,0}=\frac{x}{\lambda_2} \frac{1}{1-z} V_{1,0,0}.
\end{equation}
Also, the summation of all the steady states equals $1$ ($\mathbf{V} \mathbf{1}=1$), and by that we can obtain
\begin{equation} \label{V_101_V_100_1}
V_{1,0,1}=\frac{\lambda_1 \lambda_2+\bar{\lambda}_1  \lambda_2^2+\lambda_1^2 \bar{\lambda}_2}{\bar{\lambda}_1 \lambda_2 \bar{\lambda}_2 - \lambda_2} V_{1,0,0} +\lambda_1.
\end{equation}
as one relation between the two basic states of $V_{1,0,0}$ and $V_{1,0,1}$. 
Also, by replacing (\ref{Delta_01}) into (\ref{V_101_Delta_01}) we get another relation
\begin{equation} \label{V_101_V_100_2}
V_{1,0,1}=\frac{\bar{\lambda}_1 \lambda_2^2}{\bar{\lambda}_2}(\frac{1}{1-\bar{\lambda}_1 \bar{\lambda}_2})V_{1,0,0}
\end{equation}
between $V_{1,0,1}$ and $V_{1,0,0}$. Solving (\ref{V_101_V_100_1}) and (\ref{V_101_V_100_2}) yields
\begin{equation} \label{v_100}
V_{1,0,0}=\frac{  \lambda_1 (1-\bar{\lambda}_1 \bar{\lambda}_2) \bar{\lambda}_2 \lambda_2  }{  \bar{\lambda}_1 \lambda_2^3 +\lambda_1 \lambda_2 \bar{\lambda}_2 + \bar{\lambda}_1 \bar{\lambda}_2 \lambda_2^2 + \lambda_1^2 \bar{\lambda}_2^2 },
\end{equation}
and
\begin{equation} \label{v_101}
V_{1,0,1}=\frac{  \lambda_1 \bar{\lambda}_1  \lambda_2^3  }{  \bar{\lambda}_1 \lambda_2^3 +\lambda_1 \lambda_2 \bar{\lambda}_2 + \bar{\lambda}_1 \bar{\lambda}_2 \lambda_2^2 + \lambda_1^2 \bar{\lambda}_2^2 }.
\end{equation}

If we weight the terms in (\ref{StateEquations_AoA}) corresponding to their AoA, meaning that we multiply terms $V_{A,C,B}$ by $A$, sum all of the weighted terms, and replace the values of $V_{1,0,0}$ and $V_{1,0,1}$ from (\ref{v_100}) and (\ref{v_101}), and $w$, $x$, $y$, and $z$ we can obtain the expected value of the AoA, i.e., the average AoA as in (\ref{Average_AoA}).

\section{Proof of the Theorem 2: The Average AoAI} \label{Proof of the Theorem 2: The Average AoAI}
To calculate its average, we model a joint three-dimensional Markov process.
Each state is $(AI,I,B)$, in which $AI \in \{1,2,3,\hdots\}$ is the AoIA at the end of the time slot. $B \in \{0,1\}$ is the battery being empty or not, with $0$ representing an empty battery and $1$ a full battery at the end of time slot $t$. $I \in \{1,2,3,\hdots \}$ is the AoI, at each time slot. Thus, the states are $\{1,1,0\}$, $\{1,1,1\}$, $\{2,1,0\}$, $\{2,2,0\}$, $\{2,2,1\}$, $\{3,1,0\}$, $\{3,2,0\}$, $\{3,3,0\}$, $\{3,3,1\}$, $\{4,1,0\}$  and so on. Obviously, we cannot have states such that $AI<I$. Also, we cannot have states such that $B=1, I<AI$, since if there was enough energy, then the packet in the cache should have been utilized and the AoAI and AoI should have been the same, i.e., $B=1$ only if $AI=I$. Then, the transition probability matrix for this Markov process is (\ref{TPM_AoAI}).

The first two equations are
\begin{align} 
V_{1,1,0}&=\lambda_1 \lambda_2 B_0 + \lambda_1 \bar{\lambda}_2 B_1 \label{v_110}, \\
V_{1,1,1}&=\lambda_1 \lambda_2 B_1,  \label{v_111}
\end{align}
in which $B_1=\sum_{i=1}^{\infty}V_{i,i,1}$ and $B_0=1-B_1$. Other equations yield patterns as
\begin{equation} \label{V_AI10}
V_{AI,1,0}=\lambda_1 \bar{\lambda}_2 \sum_{I=1}^{A-1} V_{A-1,I,0},
\end{equation}
\begin{equation} \label{V_AII0}
V_{AI,I,0}= \bar{\lambda}_1  \bar{\lambda}_2 V_{A-1,I-1,0}  \ , I \neq 1, I \neq AI,
\end{equation}
\begin{equation} \label{V_AII0_2}
V_{AI,I,0}=\bar{\lambda}_1  \bar{\lambda}_2 V_{AI-1,I-1,0} + \bar{\lambda}_1  \lambda_2 \sum_{i=AI}^{\infty} V_{i,I-1,0} \ , AI=I \geq 2,
\end{equation}
\begin{equation} \label{V_AII1}
V_{AI,I,1}=\bar{\lambda}_1  \lambda_2 V_{AI-1,I-1,0} + \bar{\lambda}_1  V_{AI-1,I-1,1} \ , AI=I \geq 2.
\end{equation}
Equations (\ref{v_110}) and (\ref{v_111}) yield
\begin{equation} \label{A_1}
AI_1=\lambda_1 \lambda_2 + \lambda_1 \bar{\lambda}_2 B_1.
\end{equation}
Summing the equations from pattern (\ref{V_AI10}) yield
\begin{equation} \label{I_1_A_1}
I_1 -AI_1 = \lambda_1 \bar{\lambda}_2 B_0.
\end{equation}
From (\ref{A_1}) and (\ref{I_1_A_1}) we get that 
\begin{equation} \label{I_1}
I_1=\lambda_1,
\end{equation}
By using (\ref{I_1}) (in structures of (\ref{V_AII0_2})) and all the structures (\ref{V_AI10}), (\ref{V_AII0}), (\ref{V_AII0_2}), and (\ref{V_AII1}) we can recursively write all the states based on only two basic states of $V_{1,1,0}$ and $V_{1,1,1}$.

Then, the steady states can be written against $\lambda_1$ and $\lambda_2$.
{\small
\begin{align*}
V_{1,1,0}=&V_{1,1,0}\\
V_{1,1,1}=&V_{1,1,1}\\
V_{2,1,0}=&[x] V_{1,1,0}\\
V_{2,2,0}=&[z-y] V_{1,1,0} + [-y] V_{1,1,1}+y \lambda_1\\
V_{2,2,1}=&[y] V_{1,1,0} + [\bar{\lambda}_1] V_{1,1,1}\\
V_{3,1,0}=&[x^2 + xz-xy] V_{1,1,0}+[-xy]V_{1,1,1}+xy\lambda_1 \\
V_{3,2,0}=&[xz] V_{1,1,0}\\
V_{3,3,0}=&[z^2 -2yz] V_{1,1,0} + [-2yz]V_{1,1,1} + 2yz\lambda_1\\
V_{3,3,1}=&[yz -y^2 +\bar{\lambda}_1 y ] V_{1,1,0} + [-y^2 + \bar{\lambda}_1^2]V_{1,1,1} + y^2 \lambda_1\\
V_{4,1,0}=&[x^3 + 2x^2 z - x^2y +xz^2 -2xyz ]\\
&V_{1,1,0}+[-x^2y-2xyz]V_{1,1,1}+x^2y\lambda_1 + 2xyz\lambda_1 \\
V_{4,2,0}=&[x^2z+xz^2-xyz] V_{1,1,0}+[-xyz]V_{1,1,1}+xyz\lambda_1\\
V_{4,3,0}=&[xz^2] V_{1,1,0} \\
V_{4,4,0}=&[z^3 -3yz^2] V_{1,1,0} + [-3yz^2] V_{1,1,1} + 3yz^2 \lambda_1\\
V_{4,4,1}=&[yz^2 -2y^2z +yz\bar{\lambda}_1 -y^2 \bar{\lambda}_1 +y \bar{\lambda}_1^2 ] V_{1,1,0} +\\
&[-2y^2z -y^2 \bar{\lambda}_1 + \bar{\lambda}_1^3 ]V_{1,1,1} + 2y^2z \lambda_1+y^2 \bar{\lambda}_1 \lambda_1\\
V_{5,1,0}=&[x^4 + 3x^3 z - x^3y +3x^2z^2 -3x^2yz + xz^3 - 3xyz^2 ]\\
&V_{1,1,0}+[-x^3y-3x^2yz-3xyz^2]V_{1,1,1}\\
&+x^3y\lambda_1 + 3x^2yz\lambda_1+3xyz^2\lambda_1 \\
V_{5,2,0}=&[x^3 z+ 2x^2 z^2 - x^2yz +xz^3 -2xyz^2 ] V_{1,1,0}+\\
&[-x^2yz-2xyz^2]V_{1,1,1}+x^2yz\lambda_1 + 2xyz^2\lambda_1\\
V_{5,3,0}=&[x^2z^2+xz^3-xyz^2] V_{1,1,0} + [-xyz^2]V_{1,1,1 }+xyz^2\lambda_1\\
V_{5,4,0}=&[xz^3] V_{1,1,0}\\
V_{5,5,0}=&[z^4 -4yz^3] V_{1,1,0} + [-4yz^3]V_{1,1,1} + 4yz^3 \lambda_1\\ 
V_{5,5,1}=&[yz^3 -3y^2z^2 +yz^2\bar{\lambda}_1 -2y^2z \bar{\lambda}_1 +\bar{\lambda}_1^2 yz -\bar{\lambda}_1^2 y^2 + \bar{\lambda}_1^3 y ]\\ &V_{1,1,0}+ [-3y^2z^2 -2y^2z \bar{\lambda}_1 -\bar{\lambda}_1^2 y^2 +  \bar{\lambda}_1^4 ]V_{1,1,1}\\
&+ 3y^2z^2 \lambda_1+2 y^2 z  \lambda_1 \bar{\lambda}_1 + y^2 \bar{\lambda}_1^2 \lambda_1\\ \numberthis \label{StateEquations_AoAI}
\end{align*}
}

The summation of all the states should be equal to $1$
\begin{align*}
&[\frac{y+\lambda_1z}{(1-z)\lambda_1}  - \frac{y^2}{(1-z)^2 \lambda_1} + \frac{x-y}{(1-z)(1-x-z)} ] V_{1,1,0}\\
+&[    \frac{\bar{\lambda}_1}{\lambda_1}  - \frac{y^2}{(1-z)^2 \lambda_1}   -\frac{y}{(1-z)(1-x-z)}]V_{1,1,1}\\
+&[    \frac{ y^2}{(1-z)^2}   +    \frac{\lambda_1y}{(1-z)(1-x-z)}]=1 \numberthis \label{FirstRelationBetweenTheVsForAoAI}
\end{align*}
This is one relation between $V_{1,10}$ and $V_{1,1,1}$ and the other relation can be obtained from (\ref{v_111}):
{\small
\begin{align*}
&\left(\frac{wy(1-z)-wy^2}{(1-z)^2\lambda_1}\right)\times V_{1,1,0}+ \left(\frac{w\bar{\lambda}_1}{\lambda_1} - \frac{wy^2}{(1-z)^2 \lambda_1}+w-1 \right) \times\\
&V_{1,1,1} =- \frac{w\lambda_1 y^2}{(1-z)^2\lambda_1}. \numberthis \label{SecondRelationBetweenTheVsForAoAI}
\end{align*}}
Having these two relations, we obtain 
\begin{equation} \label{v_110_}
V_{1,1,0}=\frac{\lambda_1 (\lambda_1^2  \bar{\lambda}_2 +\lambda_2) \bar{\lambda}_2 \lambda_2}{
\lambda_1^2 \bar{\lambda}_2^2 + \lambda_2^2 + \lambda_1 \lambda_2(1 - 2 \lambda_2)},
\end{equation}
and
\begin{equation} \label{v_111_}
V_{1,1,1}= \frac{\bar{\lambda}_1 \lambda_1 \lambda_2^3}{
\lambda_1^2 \bar{\lambda}_2^2 + \lambda_2^2 + \lambda_1 \lambda_2(1 - 2 \lambda_2)}.
\end{equation}
If we multiply the states equations (\ref{StateEquations_AoAI}) by their AoAI, i.e. multiply $V_{AI,I,B}$ by $AI$, and sum them together, we get the average:

$\overline{AI}=[\frac{x - y}{(1 - x - z)^2} + \frac{y + 1}{(1 - z)^2} + \frac{y}{(1 - z) \lambda_1^2} + \frac{y}{(1 - z)^2  \lambda_1} - \frac{y^2}{(1 - z)^2  \lambda_1^2} - \frac{2  y^2}{(1 - z)^3 \lambda_1} + \frac{x}{(1 - z)^2  (1 - x - z)} + \frac{z  (x - y)}{(1 - z)  (1 - x - z)^2} - \frac{x  y}{(1 - x - z)  (1 - z)^3} + \frac{y  (z - 2)}{(1 - z)^3}]V_{1,1,0}
+[- \frac{y}{(1 - x - z)^2} + \frac{y}{(1 - z)^2} - \frac{y^2}{(1 - z)^2 \lambda_1^2} - \frac{2 y^2}{(1 - z)^3 \lambda_1} + \frac{1}{\lambda_1^2} - \frac{x y}{(1 - x - z) (1 - z)^3} + \frac{y (z - 2)}{(1 - z)^3} - \frac{y z}{(1 - x - z)^2 (1 - z)}]V_{1,1,1}
+[\frac{\lambda_1 y}{(1 - x - z)^2} - \frac{\lambda_1 y}{(1 - z)^2} + \frac{y (\lambda_1 (2 - z) + 2 y)}{(1 - z)^3} + \frac{y^2}{(1 - z)^2 \lambda_1} + \frac{\lambda_1 x y}{(1 - x - z) (1 - z)^3} + \frac{\lambda_1 y z}{(1 - x - z)^2 (1 - z)]}
$\\

If we replace (\ref{v_110_}) and (\ref{v_111_}) and also $w$, $x$, $y$, and $z$, and simplify, we obtain the average AoAI (\ref{Average_AoAI}).

\end{appendices}

\bibliographystyle{ieeetr}
\bibliography{bibliography.bib}

\end{document}